\documentclass[a4paper, 11pt]{article}

\usepackage{calc, amsmath, amsthm, a4, latexsym, amssymb, color, url}
\usepackage{graphicx}
\usepackage{psfrag}

\psfrag{V}{\small $V$}
\psfrag{U}{\small $U$}
\psfrag{Z}{\small $Z$}

\definecolor{comcolor}{rgb}{0.9,0.3,0.3}
\definecolor{starcolor}{rgb}{0.3,0.3,0.9}
\definecolor{hscolor}{rgb}{0.9,0.6,0.5}

\newtheorem{thm}{Theorem}[section]
\newtheorem{lemma}[thm]{Lemma}
\newtheorem{corollary}[thm]{Corollary}
\newtheorem{prop}[thm]{Proposition}

\theoremstyle{definition}
\newtheorem{defn}[thm]{Definition}

\newtheorem{rem}[thm]{Remark}

\newcommand{\be}{\begin{equation}}
\newcommand{\ee}{\end{equation}}
\newcommand{\ba}{\begin{array}}
\newcommand{\ea}{\end{array}}
\newcommand{\bal}{\begin{aligned}}
\newcommand{\eal}{\end{aligned}}

\newcommand{\R}{\mathbb{R}}

\newcommand{\N}{\mathbb{N}}

\newcommand{\E}{\mathbb{E}}

\renewcommand{\P}{\mathbb{P}}

\newcommand{\1}{1\hspace{-0.098cm}\mathrm{l}}

\newcommand{\dd}{{\text{d}}}

\newcommand{\bz}{{\overline{z}}}
\newcommand{\br}{{\overline{r}}}
\newcommand{\bc}{{\overline{c}}}
\newcommand{\bZ}{{\overline{Z}}}
\newcommand{\bmu}{{\overline{\mu}}}
\newcommand{\bLambda}{{\overline{\Lambda}}}

\DeclareMathOperator{\Bin}{Bin}
\DeclareMathOperator{\Hyp}{Hyp}

\begin{document}

\title{The seed bank coalescent with simultaneous switching}

\author{Jochen Blath\footnote{TU Berlin, blath@math.tu-berlin.de}, Adri\'an Gonz\'alez Casanova\footnote{UNAM, adriangcs@matem.unam.mx}, Noemi Kurt\footnote{TU Berlin, kurt@math.tu-berlin.de}, Maite Wilke-Berenguer\footnote{Ruhr-Universit\"at Bochum, maite.wilkeberenguer@rub.de}} 

\maketitle

\begin{abstract}
We introduce a new Wright-Fisher type model for seed banks incorporating ``simultaneous switching'', which is motivated by recent work on microbial dormancy (\cite{LJ11}, \cite{SL18}). We show that the simultaneous switching mechanism leads to a new jump-diffusion limit for the scaled frequency processes, extending the classical Wright-Fisher and seed bank diffusion limits. We further establish a new dual coalescent structure with multiple activation and deactivation events of lineages. While this seems reminiscent of multiple merger events in general exchangeable coalescents, it actually leads to an entirely new class of coalescent processes with unique qualitative and quantitative behaviour. To illustrate this, we provide a novel kind of condition for coming down from infinity for these coalescents 
using recent results of Griffiths \cite{G14}.
\end{abstract}

\section*{Introduction}

The evolutionary consequences of dormancy resp.\ the presence of a seed bank in a population are currently an active topic both in the biologically as well as the mathematically oriented population genetics communities (e.g.\ \cite{KKL01}, \cite{V04}, \cite{T11}, \cite{BGCKS12}, \cite{BGCKW16}, \cite{DHP16}, \cite{MKTZ17}, \cite{SL18}). Indeed, seed banks are believed to strongly affect the interplay of classical evolutionary forces such as genetic drift, selection and migration; and mathematical (toy-) models and inference tools for seed banks are currently being developed (\cite{BBKWB18+}). We refer to \cite{SL18} for a comprehensive overview and many further references. However, at present there seems to be a whole range of more or less natural ways to model a seed bank, and different models predict different qualitative behaviour (e.g.\ ``weak'' vs.\ ``strong'' seed banks, cf.\ \cite{BGCKW16}, \cite{KKL01}, \cite{ZT12}). Moreover, for several important scenarios, adequate mathematical models are still missing entirely.

In \cite{LJ11}, Lennon and Jones discuss various biological mechanisms (with a focus on microbial species) that lead to the initiation of dormancy and the resuscitation of dormant organisms. In particular, they distinguish between {\em spontaneous switching} and {\em simultaneous switching}, where the first mechanism describes the spontaneous initiation of dormancy in a single microbe, independent of the state of rest of the population, while the latter describes the simultaneous initiation of dormancy in a whole fraction of the population, say in response to an environmental cue (such as changes in temperature, or availability of resources). This mechanism is thus also known as responsive switching. A similar distinction can be made for the resuscitation from a dormant state (individually vs.\ simultaneously due to a trigger event).

The first mechnism, {\em spontaneous switching} has been incorporated in \cite{BEGCKW15}, \cite{BGCKW16} into a population model related to Wright's two island model (\cite{W51}, \cite{KZH08}), where the islands correspond to the active and the dormant sub-population (with the distinguishing  feature that reproduction is blocked in the dormant part). Here, spontaneous switching events correspond to what one would traditionally call migration between the two sub-populations. Yet, rather surprisingly, there are several qualitative and quantitative differences between the resulting seed bank diffusion limit and the classical two island diffusion, see e.g.\ \cite{BBGCWB18+}. Both models have an interesting ancestral dual process, namely the seed bank coalescent (see \cite{BEGCKW15} and also \cite{LM15} for a similar structure arising in the context of peripatric speciation models) and the well-known structured coalescent (cf.\ e.g.\ \cite{H94, T88, N90}). While the structured coalescent is well-established, the seed bank coalescent is new and still under investigation, and inference tools are currently being developed (\cite{BBKWB18+}. 

However, {\em simultaneous switching} seems to have not been incorporated in Wright-Fisher type seed bank models so far, although it appers to be an important mechanism for seed bank dynamics (\cite{LJ11}). It is the purpose of this paper to provide a first (toy-)model for this scenario and to analyse its scaling limit and dual ancestral process. We will see below that the resulting coalescent process, called the seed bank coalescent with simultaneous switching, is a new mathematical object with unique properties. As in the classical seed bank coalescent, lines can be either active or dormant, and the coalescence dynamics regarding the active lines are similar to a Kingman coalescent, while dormant lines are blocked from coalescence. However, lines can switch their status from active to dormant and vice versa simultaneously according to some driving Poisson measure, so that multiple lines can become active or dormant at a time.  This feature extends the individual switching of the seed bank coalescent and leads to new qualitative behaviour. The switching of multiple lines at the same time is reminiscent of multiple merger events in Lambda-coalescents (\cite{S99}, \cite{P99}, \cite{DK99}), yet leads to different tree structures, which is reflected in a new type of criterion for ``coming down from infinity'', interestingly involving arguments from rather elegant recent work by Griffiths \cite{G14}. 

The paper is organized as follows. In Section \ref{sn:forward_model}, we define two variants of seed bank models incorporating simultaneous switching and show that their corresponding allele frequency processes converge to a certain jump-diffusion limit (the seed bank diffusion with jumps), under a classical re-scaling similar to the Wright-Fisher model and the Wright Fisher diffusion. In Section \ref{sn:seedbank_coalescent}, we first define the seed bank coalescent with simultaneous switching and show that it is the moment dual to the seed bank diffusion with jumps, thus describing the ancestry of samples from this model. We will then discuss absorption probabilities and long-term behavior of the diffusion with the help of this dual, before investigating conditions for the coming down from infinity.

\section{The forward model and its scaling limit}
\label{sn:forward_model}

In this section we present forward in time population models with seed bank, allowing for spontaneous as well as simultaneous switching. We proceed in two steps, first presenting a model with a fixed fraction of individuals involved in simultaneous switching event, later generalizing to random numbers. The second model is a generalization of the first one, and most of our results will be stated for this general case. However, for simplicity of presentation, we start with the easier situation of fixed switching size. 

Consider a haploid population of fixed size $N$ of active individuals reproducing in discrete non-overlapping generations $k=0,1,...$ Assume that individuals carry a genetic type from some type-space $E$ (we will later pay special attention to the bi-allelic setup, say $E=\{a, A\}$, for the forward model). Further, assume that the population also sustains a \emph{seed bank} of constant size $M$, which consists of the dormant individuals. For simplicity, we will sometimes refer to the $N$ `active' individuals as `plants' and to the $M$ dormant individuals as `seeds' (even if they are typically microorganisms).

\subsection{Model A: Simultaneous switching of fixed size}
Fix $c, \bc >0$, which will describe the (small) number of individuals affected by spontaneous switching events, and fix $z,\bz\in [0,1]$ as parameters for the large simultaneous migration events. The model is then defined with the help of three types of events. For simplicity of notation, we assume first that $c, \bc zN$ and $\bz M$ are natural numbers (otherwise Gauss-brackets could to be introduced into the definition in a suitable manner). We assume that in each generation, reproduction is governed by one of the following three events:

\begin{itemize}
\item[S] \emph{Spontaneous switching (small-scale migration event of size $o(N)$) between active and dormant:} 
For the new active generation, $N-c$ active indiviuals are obtained by multinomial sampling from the previous active generation. The remaining $c$ active slots are filled by sampling (without replacement) $c$ types independently and uniformly from the seed bank types of the previous generation. 
For the new dormant generation, $M-\bc$ dormant individuals chosen uniformly at random simply stay in the seed bank, and the remaining slots are filled up by $\bc$ new ones via multinomial sampling from active individuals in the previous generation.

\item[F] \emph{Simultaneous switching (large-scale migration of size $O(N)$) from dormant to active, ``forest fire'':}
For the new active generation, $(1-z)N$ active individuals are obtained by multinomial sampling from the previous active generation.
The remaining $zN$ active slots are filled by sampling (without replacement) $zN$ types independently and uniformly from the seed bank types of the previous generation. The seed bank stays as it is.

\item[D] \emph{Simultaneous switching (large-scale migration of size $O(N)$) from dormant to active, ``drought'':}
The $N$ active individuals in the next generation are produced by multinomial sampling from the active individuals in the previous generation. 
For the new seed bank generation, $\bz M$ dormant individuals from the previous generation are replaced by new dormant individuals obtained by multinomial sampling from the previous active generation. The remaining $(1-\bz)M$ dormant individuals stay in the seed bank.
\end{itemize}

Thus, in each of the three cases, we have again $N$ active and $M$ dormant individuals in the next generation. This assumption of fixed population sizes is common in population genetics and in particular in Wright-Fisher type models. A situation in which fluctuations in population sizes are allowed will be investigated in future work. 

Note that in mechanism $F$ one needs to choose $z$ such that $zN\leq M.$ We denote by $R_k$ a random variable taking values in $\{S, F, D\}$ which determines the type of event that happens in generation $k.$ It is clear that in order to get a non-trivial scaling limit, large-scale migration events have to be rare, while small scale migration should be ``typical''. Here, the sequence $(R_k)_{k\in \N}$ will be chosen to be iid and independent of the previous random mechanisms, such that 
\begin{equation}
\label{eq:R_k_def}
\P(R_k=F)=O(1/N), \, \, \,  \P(R_k=D)=O(1/N), \, \mbox{ and } \, \P(R_k=S)=1-O(1/N).
\end{equation}
As a result, in the limit as $N\to\infty,$ simultaneous switching events can be expected to occur according to a Poisson process of finite rate. 

\begin{rem}
Our above model is a generalisation of the seed bank model from \cite{BGCKW16} by additionally introducing the simultaneous switching events. However, note that also the spontaneous switching mechanism was defined slightly differently in the above paper, where the event $S$ was replaced with the following: 
\begin{itemize}
\item[S'] \emph{Symmetric spontaneous switching:} 
For the new active generation, $N-c$ active individuals are obtained by multinomial sampling from the previous active generation. The remaining $c$ active slots are filled by sampling (without replacement) $c$ types independently and uniformly from the seed bank types of the previous generation. 
For the new dormant generation, precisely these $c$ types are replaced by $c$ new ones via  multinomial sampling from active individuals in the previous generation.
\end{itemize}

The advantage of working with $S$ instead of $S'$ is the fact that spontaneous migration from active to dormant and from dormant to active are now decoupled, and in particular, one may choose to have small migration events only in one direction (by setting either $c$ or $\bc$ equal to 0). Most of our results are true for both S and S', and the proofs immediate, by choosing $\bc=c$ in all the statements. This is due to the fact that in the limit $N\to\infty$ it doesn't matter if precisely the types that have been selected by multinomial sampling are being replaced themselves or not, which is the only difference.
\end{rem}

\subsection{Model B: Simultaneous switching of random size}
Model A can be extended to include large migration events of varying size. For $N\in\N,$ fix probability measures $\mu_N, \bmu_N$ on 
$$
I^N:=\Big\{0, \frac 1N, \frac 2N, \dots,\frac{N-1}{N}, 1\Big\} \quad \mbox{ resp. } \quad  I^M:=\Big\{0, \frac 1M, \frac 2M, \dots, \frac{M-1}{M}, 1\Big\}.
$$ 
In case $M<N$ choose $\mu_N$ such that $zN\leq M.$ Let $(Z_k)_{k\in\N}$ denote a sequence of iid random variables with distribution $\mu_N$ and $(\bZ_k)_{k\in\N}$ denote a sequence of iid random variables with distribution $\bmu_N$. 

Again, reproduction will be governed by three events $S, F, D$ as before, which are selected by a sequence of random variables $\{R_k\}$ in an iid fashion as before. The event $S$ is precisely the same as before, but the events $F$ and $D$ contain additional randomness.

Indeed, whenever $R_k=F,$ the fraction of dormant individuals becoming active in the F-event is given by the random number $Z_k,$ instead of the constant $z$, and whenever $R_k=D,$ the fraction of dormant individuals replaced by active offspring is given by the random number $\bZ_k$ (instead of $\bz$). Otherwise, the process is defined exactly the same as in model A.  Note that model A is contained in model B as a special case with the specific choices $\mu_N=\delta_z$ and $\bmu_N=\delta_{\bz}$ for some fixed $z,\bz$, and $(R_k)_{k\in \N}$ as in \eqref{eq:R_k_def}. However, the additional randomness in $F$ and $D$ may also require a different distribution of the $(R_k)_{k\in \N}$ in order to get a reasonable scaling limit. Below, we will give a condition jointly for the measures $\mu_N, \bmu_N,$ and the probabilities of occurrence of large events, which allow infinite rates for large migration events in the limit and still leads to a well-defined limiting model.

\subsection{The allele frequency processes}

From the above models, their allele frequency processes can be derived in the usual way.

\begin{defn}[Forward type configuration process]
\label{def:forward}
Fix population size $N \in \N$, seed bank size $M\in\N$, genetic type space $E$ and parameters as in the definition of the models A resp.\ B above. Given initial type configurations $\xi_0 \in E^{N}$ and $\eta_0 \in E^M$, denote by
$$
\xi_k:=\big(\xi_k(i)\big)_{i \in \{1,...,N\}}, \quad k \in \N, 
$$
the random genetic type configuration in $E^{N}$ of the active population in generation $k$ (obtained from the above mechanism), and denote by 
$$
\eta_k:=\big(\eta_k(j)\big)_{j \in \{1,...,M\}}, \quad k \in \N, 
$$
correspondingly the genetic type configuration of the dormant population in $E^{M}$.
We call the discrete-time Markov chain $(\xi_k, \eta_k)_{k \in \N_0}$ with values in $E^{N} \times E^{M}$ the \emph{type configuration process} of the \emph{Wright-Fisher model with geometric seed bank component}.
\end{defn}

We now specialise to the bi-allelic case $E=\{a,A\}$ and define the frequency processes of $a$ alleles in the active population and in the seed bank. 

\begin{defn}[Forward frequency process, biallelic case]
\label{def:forward}
With the above notation and condition, define the discrete-time Markov chain $(X^N_k, Y^M_k)_{k\in\N_0}$ on $I^N\times I^M,$ by
\begin{equation}
\label{eq:frequency_chains}
X_k^N:= \frac 1N \sum_{i=1}^N {\bf 1}_{\{\xi_k(i)=a\}} \quad \mbox{ and } \quad Y_k^M:= \frac 1M \sum_{j =1}^M {\bf 1}_{\{\eta_k(j)=a\}}, \quad k \in \N_0,
\end{equation}
Denote by $\P_{x,y}$ the initial distribution under which $(X^N, Y^M)$ starts in $(x,y), \P_{x,y}$-a.s., i.e.\ 
$$
\P_{x,y}(\,\cdot \,):= \P\big(\, \cdot \, \big| X^N_0 = x,\,Y^M_0=y\big) \quad \mbox{ for } \quad (x,y) \in I^N \times I^M
$$
(with analogous notation for the expectation $\E_{x,y}$). 
\end{defn}

Our next aim is to characterise the corresponding time-homogeneous transition probabilities. To this end, we introduce auxiliary random variables in a similar fashion as in \cite{BGCKW16}. For a fixed time $k\in \N$, let
\begin{itemize}
\item $T$ be the number of active individuals that are offspring of a dormant $a$-individual,
\item $U$ be the number of active individuals that are offspring of an active $a$-individual,
\item $V$ be the number of dormant individuals that are offspring of an active $a$-individual,
\item and $W$ be the number of dormant individuals that are offspring of a dormant $a$-individual.
\end{itemize}
With this notation, if $X^N_0=x, Y^M_0=y$ $\P$-almost surely, we have the representation
\be 
\label{eq:freq_representation}
X^N_1=\frac{1}{N}(T+U) \mbox{ and } Y^M_1=\frac{1}{M}(V+W).
\ee
According to our construction, these random variables are all independent. Of course, the distributions depend on the type of event, chosen by $R_1$, and on the choice of model A or B. In model A, the distributions are given in Table \ref{table}, where $\Hyp_{n,m,k}$ denotes the hypergeometric distribution with parameters $n,m,k\in\N,$ and $\Bin_{n,p}$ is the binomial distribution with parameters $n\in \N$ and $p\in[0,1]$. The transitions from $(X_k^N,Y_k^N)$ to $(X_{k+1}^N, Y_{k+1}^M)$ can be described analogously. In model B, conditional on the sequences $(Z_k)_{k\in\N}, (\bZ_k)_{k\in\N}$, the random variables can be constructed in a similar fashion.

\begin{table}
\begin{center}
\begin{tabular}{|c||l|l|l|l|}
\hline
$R_k$&$T$&$U$&$V$&W\\
\hline
\hline
$S$&$\sim \Hyp_{c,M,yM}$& $\sim \Bin_{N-c,x}$& $\sim \Bin_{\bc, x}$& $\equiv yM-Z$\\
\hline
$F$&$\sim \Hyp_{zN,M,yM}$&$\sim \Bin_{1-zN,x}$& $\equiv 0$&$\equiv yM$\\
\hline
$D$& $\equiv 0$& $\sim\Bin_{N,x}$& $\sim \Bin_{\bz M,x}$& $\sim\Hyp_{(1-\bz) M, M, yM}$\\
\hline
\end{tabular}
\caption{Distribution of auxiliary random variables under events $S, F, D$}
\label{table}
\end{center}
\end{table}

\subsection{Limiting generators of the frequency processes}

Here, we follow the usual scaling limit paradigm in population genetics, where it is assumed that the population size $N$ tends to $\infty$, and simultanoeously time is measured on a macrsoscopic scale increasing also with $N$. Since in our case we have populations (of size $N$ and $M$ each), we first assume that the active and the dormant population keep the same {\em relative size}, that is we set $M=M(N):=\lfloor N/K \rfloor$, for some suitable constant $K>0$, as $N \to \infty$. The following arguments follow the standard machinery for the convergence of Markov processes, as elaborated e.g.\ in \cite{EK86}, and thus we focus on the crucial steps and computations.

We begin with the scaling limit in the case of fixed simultaneous switching size, i.e.\ model A. We can define the discrete generator of the process $(X_{\lfloor D_{N,M}t\rfloor}, Y_{\lfloor D_{N,M}t\rfloor})_{t\geq 0}$ on time scale $D_{N,M}$ acting on suitable functions $f$ (e.g.\ $C^2([0,1] ^2)$ by
\[
A^{N,M}f(x,y)=D_{N,M}\E_{x,y}\big[f(X^N_1, Y^M_1)-f(x,y)\big], \quad (x,y) \in I^N \times I^M.
\]
With some experience, it is not hard to guess the shape of the limiting process. We know from \cite{BGCKW16} that the frequent small events $S'$ lead to the seed bank diffusion with migration rates $c, cK$, and it is easy to see that this is still the case for $S$, however this time with migration rates $c, \bc K$. The much rarer $F$-event leads to a jump of size $z(y-x)$ in the active population, and a $D$-event leads to a jump of size $\bz(x-y)$ in the dormant population.  

\medskip

To make this rigorous, we assume that for every $N\in\N$ the random variables $R_k=R^N_k, k\in\N$ which determine the jump types are iid and such that as $N\to\infty$
\be
\label{eq:R_split}
\P(R_k=F)=\frac{r_F}{N}(1+o(1)), \quad \P(R_k=D)=\frac{r_D}{N}(1+o(1))
\ee
for some $0\leq r_F,r_D<\infty$, and 
\be
\label{eq:R_split_2}
\P(R_k=S)=1-\P(R_k=F)-\P(R_k=D).
\ee
Observe that $I^N \times I^M \subset [0,1]^2$ and for each $(x,y) \in [0,1]^2$ denote by $\pi_N(x,y)$ the canonical projection 
on $I^N \times I^M$, mapping $(x,y)$ to the closest point in $I^N \times I^M$ with coordinates smaller than $x$ resp.\ $y$.

\begin{thm}[Limiting generator in model A]
\label{thm:conv_fixed}
Under our above assumptions, we obtain with the choice $D_{N,M}=N$ that
\[
\lim_{N\to\infty} \sup_{(x,y) \in [0,1]^2}\big|A^{N,M}f(\pi_N(x,y))-Af(x,y)\big| =0
\]
for all $f\in C^{2}([0,1]^2)$, where 
\begin{align*}Af(x,y)=&r_Ff(x+z(y-x),y)+r_Df(x,y+\bz(x-y))-(r_F+r_D)f(x,y)\\
&+c(y-x)\frac{\partial }{\partial x}f(x,y)+\bc K(x-y)\frac{\partial }{\partial y}f(x,y)+\frac{1}{2}x(1-x)\frac{\partial^2 }{\partial x^2}f(x,y).\end{align*}
\end{thm}

\begin{rem}
If mechanism $S'$ is assumed instead of $S$ in model A, the result holds with $\bc$ replaced by $c.$
\end{rem}

\begin{proof}

By standard arguments, it is sufficient to prove the stated convergence for polynomials $f(x,y)=x^ny^m, n,m\in\N_0,$ on $[0,1]^2$, since polynomials are dense in $C^{2}([0,1]^2)$. Using \eqref{eq:R_split} and \eqref{eq:R_split_2}, we can split according to the different values of $R_1$ to obtain (for $N$ large enough)
\begin{align}\label{eq_split}
\E_{x,y}\big[f(X^N_1,Y^M_1)-&f(x,y)\big]\\  
=&\E_{x,y}\big[f(X^N_1,Y^M_1)-f(x,y)\,|\,R_1=S\big]\P_{x,y}(R_1=S) \nonumber\\
&+\E_{x,y}\big[f(X^N_1,Y^M_1)-f(x,y)\,|\,R_1=F\big]\P_{x,y}(R_1=F) \nonumber \\&
+\E_{x,y}\big[f(X^N_1,Y^M_1)-f(x,y)\,|\,R_1=D\big]\P_{x,y}(R_1=D)\nonumber
\end{align}
for all $(x,y) \in  I^N \times I^M$. In \cite{BGCKW16}, Proposition 2.4, it was shown that
\begin{align*}
\lim_{N\to\infty} \sup_{(x,y)\in[0,1]^2}N &\E_{\pi_N(x,y)}\big[f(X^N_1,Y^M_1)-f(x,y)\,|\,R_1=S'\big] \\
=&c(y-x)\frac{\partial }{\partial x}f(x,y)+c K(x-y)\frac{\partial }{\partial y}f(x,y)+\frac{1}{2}x(1-x)\frac{\partial^2 }{\partial x^2}f(x,y),
\end{align*}
uniformly for all $(x,y) \in [0,1]^2$. The case for $R_1=S$ works similarly and leads to the same result, except that in the coefficient of  $\frac{\partial }{\partial y}f(x,y)$, the constant $c$ is replaced by $\bc$. We skip the somewhat tedious calculations and refer to the Appendix of \cite{BGCKW16} instead. Since $\P_{x,y}(R_1=S)$ converges for $N\to\infty$ to 1 uniformly in $x$ and $y,$ we obtain the desired convergence of the first summand in \eqref{eq_split}.

Consider now $R_1=F,$ the case $R_1=D$ works similarly. By construction, we have for $f(x,y)=x^ny^m$, using \eqref{eq:freq_representation} and Table \ref{table},
\begin{align*}
\E_{x,y}\big[f(X^N_1,Y^M_1)-&f(x,y)\,\big|\,R_1=F\big]\\
=&\E_{x,y}\big[(X^N_1)^n(Y_1^M)^m-x^ny^m\,\big|\,R_1=F\big]\\
=&\frac{1}{N^nM^m}\E_{x,y}\big[(T+U)^n(V+W)^m1_{\{R_1=F\}}\big]-x^ny^m\\
=&\frac{1}{N^n}y^m\E_{x,y}\big[(T+U)^n1_{\{R_1=F\}}\big]-x^ny^m.
\end{align*} 

We claim that for all $n\in\N$, on $\{R_1=F\}$
\be 
\label{eq:jump}
\frac{1}{N^n}\E_{x,y}\big[(T+U)^n\big]=(x+z(y-x))^n+C_N(x,y),
\ee
with $\sup_{x,y\in[0,1]}C_N(x,y)\leq N^{-1}.$ Then the result follows, since $N\P_{x,y}(R_1=F)\to r_F$ as $N\to\infty$ uniformly in $x$ and $y.$

To prove \eqref{eq:jump}, observe that
\[
\E_{x,y}\big[(T+U)^n\big]=\big(\E_{x,y}[T]+E_{x,y}[U]\big)^n+\E_{x,y}\big[(T-\E[T])^n\big]+\E_{x,y}\big[(U-\E_{x,y}[U])^n\big]+R_{n,x,y}(T,U),
\]
where $R_{n,x,y}(T,U)$ consists of mixed terms of the form 
$$
C\big(\E_{x,y}[T]+\E_{x,y}[U]\big)^k\E_{x,y}\big[(T-\E_{x,y}[T])^l\big]\E_{x,y}\big[(U-\E_{x,y}[U])^m\big],
$$ 
with $k,l,m\leq n-1, k+l+m=n$ and combinatorial prefactors $C$ depending only on $k,l$ and $m.$ Note that 
$$
\E_{x,y}[T]+\E_{x,y}[U]=(zy+(1-z)x)N=(x+z(y-x))N.
$$ 
We are thus done once we prove that the $n$th centered moments of $T$ and $U$ are of order at most $N^{n-1}, n\in\N,$ uniformly in $x$ and $y.$ For the first two centered moments of $T$ we have $\E_{x,y}[T-\E_{x,y}[T]]=0$ and 
$$
\E_{x,y}\big[|T-\E_{x,y}[T]|^2\big]=\mathbb{V}_{x,y}(T)=Nzy(1-y)\frac{M-zN}{M-1}\leq N.
$$ 
For $n\geq 3$ we have
\[
\E_{x,y}\big[|T-\E_{x,y}[T]|^n\big]=\E_{x,y}\big[|T-\E_{x,y}[T]|^{n-1}|T-\E_{x,y}[T]|\big]\leq (N+1)\E_{x,y}\big[|T-\E_{x,y}[T]|^{n-1}\big],
\]
since $T$ is hypergeometric with values in $\{0,...,M\},$ and thus trivially $|T-\E_{x,y}[T]|\leq M+1$. By induction, $\E_{x,y}[|T-\E_{x,y}[T]|^{n-1}]\leq (M+1)^{n-2}, n\geq 2.$ This implies that 
$$
\sup_{x,y\in[0,1]}\big|\E_{x,y}\big[(T-\E_{x,y}[T])^n\big]\big|\leq \sup_{x,y\in[0,1]} \E_{x,y}\big[|T-\E_{x,y}[T]|^n\big]\leq (KN)^{n-1}+O(N^{n-2}).
$$ 
for all $n\geq 1$ (recall $N=KM$). Similar considerations hold for $U,$ which is binomial.
\end{proof}


In model B we make the following assumption. Let $(r_N)_{N\in\N}$ and  $(\br_N)_{N\in\N}$ denote sequences of nonnegative numbers such that $(r_N/N)_{N\in\N}$ and $(\br_N/N)_{N\in\N}$ converge to 0 as $N\to\infty.$ Assume that there exist measures $\mu,\bmu$ on $[0,1]$ with 
\be
\label{eq:con_finite_meas}
\int_{[0,1]}z \mu(\dd z)<\infty\quad \mbox{ and } \quad \int_{[0,1]}\bz \bmu(\dd \bz)<\infty
\ee
such that weakly,
\be
\label{eq:cond_conv_meas}
\lim_{N\to\infty}r_N\mu_N=\mu 
\ee
and analogously for $\br_N, \bmu_N, \bmu.$ Observe that in particular $\mu,\bmu$ need not be finite measures. 

We now assume that each sequence $(R^N_k)_{k\in \N}, N\in\N$ is iid such that as $N\to\infty$
\be
\label{eq:rN_cond_1}
\P(R^N_k=F)=\frac{r_N}{N}(1+o(1)), \qquad \P(R^N_k=D)=\frac{\br_N}{N}(1+o(1)), 
\ee
and
\be
\label{eq:rN_cond_2}
\P(R^N_k=S)=1-\frac{r_N+\br_N}{N}(1+o(1)).
\ee

\begin{thm}[Limiting generator in model B]
\label{thm:conv_random}
Under our above assumptions, we obtain with the choice $D_{N,M}=N$ that
\[
\lim_{N\to\infty} \sup_{(x,y) \in [0,1]^2}\big|A^{N,M}f(\pi_N(x,y))-Af(x,y)\big| =0
\]
for all $f\in C^{2}([0,1]^2)$, where  
\begin{align*}Af(x,y)=&\int_{[0,1]}(f(x+z(y-x),y)-f(x,y))\mu(\dd z)\\
&+\int_{[0,1]}(f(x,y+\bz(x-y)-f(x,y))\bmu(\dd \bz)\\
&+c(y-x)\frac{\partial }{\partial x}f(x,y)+\bc K(x-y)\frac{\partial }{\partial y}f(x,y)+\frac{1}{2}x(1-x)\frac{\partial^2 }{\partial x^2}f(x,y).\end{align*}
\end{thm}

\begin{rem}
Note that in particular the functions $f(x,y)=x^ny^m, n,m\in\N_0$ are in the domain of $A$. If $n,m\geq 1,$ this follows from \eqref{eq:con_finite_meas}, if $n=0$ we have $f(x+z(y-x),y)-f(x,y)=0,$ and analogously for $m=0.$ 
\end{rem}

\begin{proof}
The proof follows from Theorem \ref{thm:conv_fixed} if we additionally show that uniformly in $(x,y)\in[0,1]^2$
\[\lim_{N\to\infty}\int_{[0,1]} [f(x+z(y-x),y)-f(x,y)]\frac{r_N}{N}\mu_N(\dd z)=\int_{[0,1]} [f(x+z(y-x),y)-f(x,y)]\mu(\dd z)\]
and 
\[\lim_{N\to\infty}\int_{[0,1]} [f(x,y+\bz(x-y))-f(x,y)]\frac{\br_N}{N}d\bmu_N(\dd \bz)=\int_{[0,1]} [f(x+ \bz(y-x),y)-f(x,y)] \bmu(\dd \bz)\]
hold and are finite, which by construction is the case if and only if the integrals on the rhs are finite for every $x,y\in[0,1],$ due to the weak convergence of measures.
By density of the monomials it is sufficient to check this for functions of the form $f(x,y)=x^ny^m,$ and because we are working on $[0,1],$ by monotonicity, it is sufficient to look at $f(x,y)=x$ and $f(x,y)=y$ (all other mixed monomials are bounded by these two). But we have
\[\sup_{(x,y)\in[0,1]^2}\int_{[0,1]} |x+z(y-x)x|\mu(\dd z)=\sup_{(x,y)\in[0,1]^2}\int_{[0,1]} z|y-x|\mu(\dd z)<\infty\]
according to the assumption $\int_{[0,1]}z \mu(\dd z)<\infty$ and $|x-y|\leq 1,$ and likewise for the other cases. This completes the proof.
\end{proof}

\begin{rem}\label{rem:Lambdas}
The condition $\int_{[0,1]}z \mu(\dd z)<\infty$ implies that $\Lambda(A):=\int_Az \mu(\dd z)$ for Borel sets $A\subseteq [0,1]$ defines a finite measure $\Lambda$ on $[0,1]$ which satisfies $\Lambda(\{0\})=0.$ On the other hand, if $\Lambda$ is a finite measure on $[0,1]$ with $\Lambda(\{0\})=0,$ then $\mu(\dd z):=z^{-1}\Lambda(\dd z)$ defines a $\sigma-$finite measure on $(0,1].$ We may extend it to $[0,1]$ by setting $\mu(\{0\})=0,$ which is no restriction, since choosing $z=0$ in the large migration mechanism has no effect. We will thus often 
work with $\Lambda$ instead of $\mu,$ and similarly with $\bLambda$ instead of $\bmu.$ The condition on $\mu$ resp. on $\Lambda$ is also necessary to define a dual process with finite rates, see Definition \ref{Block} later on. We further elaborate on this point in remark \ref{rem:singularity}.

Given a finite non-zero measure $\Lambda$ on $[0,1]$ with $\Lambda(\{0\})=0$, it is always possible to construct a sequence of probability measures $\mu_N$ on $I^N$ and a sequence $(r_N)_{N\in\N}$ such that $r_N/N\to 0$ as $N\to\infty$ and $r_N\mu_N(\dd z)\to z^{-1}\Lambda(\dd z)$ weakly. 

\end{rem}

\subsection{The seed bank diffusion with jumps and its dual process}
\label{ssn:dual}

In the previous section, we showed that the generators $A^{N,M}$ of the rescaled frequency processes $(X_{\lfloor D_{N,M}t\rfloor}, Y_{\lfloor D_{N,M}t\rfloor})_{t\geq 0}$ on time scale $D_{N,M}$ in model A resp.\ model B converge to a non-trivial Markov generator. We have not yet given an explicit jump-diffusion representation for the corresponding limiting processes $(X_t, Y_t)_{t\geq 0}$, which we now provide. We will also use $(X_t, Y_t)_{t\geq 0}$ to state the moment duality of our system below.

\begin{defn}[Seed bank diffusion with fixed-size jumps]
For $z, \bz \in (0,1)$ we call the unique strong solution $(X(t), Y(t))_{t \geq 0}$, starting in $(x, y)\in [0,1]^2$, of the initial value problem 
\begin{align}
\label{eq:system_fixed}
\text{d} X(t) & = c(Y(t) -X(t))\text{d}t + \sqrt{X(t)(1-X(t))}\text{d}B(t)\, \\[.1cm]
	          & \qquad \qquad\qquad\qquad \quad \, + r_F \int \limits_{(0,t]} \big(X(r-)+z(Y(r-)-X(r-)\big) \, N_F\big(\text{d}r\big),\notag \\[.1cm]
\text{d} Y(t) & = \bc K(X(t) -Y(t))\text{d}t + r_D \int \limits_{(0,t]} \big(Y(t-)+\bz (X(t-)-Y(t-)\big) \, 
	   N_D\big(dr\big), \notag 	   
\end{align}
with $(X(0), Y(0)) =(x,y) \in [0,1]^2$, where $(B(t))_{t\geq 0}$ is a standard Brownian motion and $N_F$ and $N_D$ are independent standard Poisson processes driving the simultaneous switching events, {\em seed bank diffusion with fixed-size jumps} $(z, \bz)$.
\end{defn}

A similar representation can be provided for model B. 

\begin{defn}[Seed bank diffusion with variable-size jumps]
For $\mu, \bmu$ as in the previous section, we call the unique strong solution $(X(t), Y(t))_{t \geq 0}$, starting in $(x, y)\in [0,1]^2$, of the initial value problem 
\begin{align}
\label{eq:system_variable}
\text{d} X(t) & = c(Y(t) -X(t))\text{d}t + \sqrt{X(t)(1-X(t))}\text{d}B(t)\,  \\[.1cm]
	          & \qquad \qquad\qquad\qquad \quad \, + \int \limits_{(0,t] \times [0,1]} \big(X(r-)+z(Y(r-)-X(r-)\big) \, {N^{\mu}_F}\big(\text{d}r,\text{d}z \big),\notag \\[.1cm]
\text{d} Y(t) & = \bc K(X(t) -Y(t))\text{d}t + \int \limits_{(0,t]\times [0,1]} \big(Y(r-)+\bz (X(r-)-Y(r-)\big) \, 
	   {N^{\bar \mu}_D}\big(\text{d}r,\text{d}\bz\big), \notag 
\end{align}
with $(X(0), Y(0)) =(x,y) \in [0,1]^2$, where $(B(t))_{t\geq 0}$ is a standard Brownian motion and $({N^{\mu}_F}(t)){t\geq 0}$ and $({N^{\bar \mu}_D}(t))_{t\geq 0}$ are independent standard Poisson point processes on $(0, \infty) \times [0,1]$ with intensity measure $\lambda(\dd t) \otimes \mu(\dd z)$ resp.\ $\lambda(\dd t) \otimes \bar \mu(\dd \bz)$ driving the simultaneous switching events of random size, {\em seed bank diffusion with variable-size jumps} with jump laws $(\mu, \bmu)$. Here, $\lambda$ denotes the Lebesgue measure on $\R.$
\end{defn}

Note that the above initial value problems are two-dimensional jump-diffusions with non-Lipschitz coefficients. Fortunately, existence and uniqueness results for such systems have recently drawn considerable interest, and we may refer e.g.\ \cite{K07, K14} or the perhaps more readily accessible \cite{BLP15} for an existence and strong uniqueness result.

With the limit thus being well-defined, under the condition that $(X^N_0, Y^M_0)_{t\geq 0}$ converge weakly to $(x,y) \in [0,1]^2,$ Theorem \ref{thm:conv_fixed} resp\  \ref{thm:conv_random} imply the weak convergence
$$
\big(X^N_{\lfloor D_{N,M}t\rfloor}, Y^M_{\lfloor D_{N,M}t\rfloor}\big)_{t\geq 0} \Rightarrow \big(X^N_t, Y^M_t\big)_{t\geq 0}
$$
on the Skorohod space of c\`adl\`ag paths, where $(X^N_t, Y^M_t)_{t\geq 0}$ is the unique strong (and strong Markov) solution to the initial value problems \eqref{eq:system_fixed} resp.\ \eqref{eq:system_variable} (see e.g.\ Theorem 19.28 of \cite{K02} or Corollary 4.8.9 of \cite{EK86}).

Before we state our envisaged moment duality, we define a suitable dual process. As usual, it turns out to be the block-counting process of the coalescent process describing the genealogy, to be defined formally in Section \ref{sn:seedbank_coalescent}.

\begin{defn}
\label{Block}
With the notation of B, we define the \emph{block-counting process of the seed bank coalescent with large migration events} $(N_t,M_t)_{t \ge 0}$ to be the continuous time Markov chain taking values in $\N_0\times\N_0$ with 
transitions
\be
\label{eq:dual_rates}
(n,m)\mapsto \begin{cases}(n-1,m+1)  &\text{ at rate }  \left(c+\int_0^1z(1-z)^{n-1}\mu(\dd z)\right)n, n\geq 1 \\
(n-k,m+k)  &\text{ at rate }  \binom{n}{k}\int_0^1 z^k(1-z)^{n-k}\mu(\dd z), 2\leq k\leq n, \\
(n+1,m-1) & \text{ at rate }  \left (\bc K +\int_0^1z(1-z)^{m-1}\bmu(\dd z)\right)m,m\geq 1\\
(n+l,m-l)  &\text{ at rate }  \binom{m}{l}\int_0^1z^l(1-z)^{m-k}\bmu(\dd z), 2\leq l\leq m, \\
(n-1,m)  & \text{ at rate } \binom{n}{2}, n\geq 2.
\end{cases}
\ee
For model A we consider the special case $\mu=\delta_z$ and $\bmu=\delta_\bz$ for some $z,\bz\in[0,1].$
\end{defn}
Denote by $\P^{n,m}$ the distribution for which $(N_0,M_0)=(n,m)$ holds $\P^{n,m}$-a.s., and denote the corresponding expected value by $\E^{n,m}$. It is easy to see that, {eventually}, 
$N_t +M_t=1$ (as $t \to \infty$), $\P^{n,m}$-a.s. for all $n,m \in \N_0$. We now show that $(N_t,M_t)_{t \ge 0}$ is the \emph{moment dual} of $(X_t,Y_t)_{t \ge 0}$.

\begin{thm}
\label{thm:dual}
For every $(x,y)\in [0,1]^2 $, every $n,m\in \N_0$ and every $t\geq 0$
\begin{equation}
\mathbb{E}_{x,y}\big[X_t^n Y_t^m\big]=\mathbb{E}^{n,m}\big[x^{N_t} y^{M_t}\big].
\end{equation}
\end{thm}
\begin{proof}
Let $f(x,y;n,m):= x^ny^m $. Applying for fixed $n,m\in\N_0$ the generator $A$ of $(X_t,Y_t)_{t \geq 0}$ to $f$ acting as a function of $x$ and $y$
 gives
\begin{align*}
Af(x,y)
=&\int_{[0,1]}[((1-z)x+zy)^ny^m-x^ny^m]\mu(\dd z)\\
&+\int_{[0,1]}[((x^n((1-z)y+z x)^m-x^ny^m]\bmu(\dd z)+c(y-x)nx^{n-1}y^m\\
&+\frac{1}{2}x(1-x)n(n-1)x^{n-2}y^m+\bc K(x-y)x^{n}my^{m-1}\\
=&\sum_{k=2}^n\binom{n}{k}\int_{[0,1]}z^k(1-z)^{n-k}\mu(\dd z)(x^{n-k}y^{m+k}-x^ny^m)\\
&+\sum_{l=2}^m\binom{m}{l}\int_{[0,1]}z^l(1-z)^{n-l}\bmu(\dd z)(x^{m+l}y^{n-l}-x^ny^m)\\
&+n\big[c+\int_{[0,1]}z(1-z)^{n-1}\mu(\dd z)\big](x^{n-1}y^{m+1}-x^ny^m)\\
&+m\big[\bc K+\int_{[0,1]}z(1-z)^{m-1}\bmu(\dd z)\big](x^{n+1}y^{m-1}-x^ny^m)\\
&+\binom{n}{2}(x^{n-1}y^{m}-x^ny^m),
\end{align*}
where we have used the binomial theorem and the observation that the summands for $k=0$ and $l=0$ disappear.
Note that the \emph{rhs} is the precisely the generator of $(N_t,M_t)_{t \ge 0}$ applied to $f$ acting as a function of $n$ and $m,$ for fixed $x,y\in[0,1].$ Hence the duality follows from standard arguments, see e.g.\ \cite{JK14}, Proposition 1.2.
\end{proof}

\begin{rem}\label{rem:indblocks}[Alternative description of the block-counting process]
From the usual perspective of coalescents, we can describe the dynamics of the dual block--counting process \ref{Block} in the following intuitive way: Every block, independently of the others, migrates at rate $c$ from active to dormant and at rate $\bc K$ from dormant to active. Every given pair of active blocks coalesces at rate 1. Moreover, at fixed, constant rate 1 a large migration event from active to dormant happens, where every active block participates with probability $z$ (chosen according to $\mu$) independently of the others. Likewise, at constant rate 1, a large migration event from dormant to active happens, where every dormant block participates at rate $\bz$ independently of the others. Note that these large migration events may result in a migration of 0 blocks (with probability $(1-z)^n$), or in a migration of 1 block (with probability $z(1-z)^{n-1}$). This description makes it clear that different blocks move independently of the others, an observation which will be useful later, when we construct couplings of the block counting process with other processes.
\end{rem}

\subsection{Long-term behaviour and fixation probabilities}
\label{ssn:fixation}

The fixation probabilities can be calculated as for the usual seed bank coalescent. For simplicity we formulate the results only for model A, they are easily generalised to model B by integrating out the $z$ and $\bz$ according to the respective measures.

Obviously, $(0,0)$ and $(1,1)$ are absorbing states for the system \eqref{eq:system_variable}. They are also the only absorbing states, since absence of drift requires $x=y,$ and for the fluctuations to disappear, it is necessary to have $x\in\{0,1\}.$

\begin{prop}
\label{prop:moment_limit}
In model A, all mixed moments of $(X_t,Y_t)_{t \ge 0}$ solving \eqref{eq:system_variable} converge to the \emph{same} finite limit depending on $x,y, c,\bc, K, z,\bz$. More precisely, for each fixed $n,m\in\N$, we have
\begin{equation}
\label{eq:moment_value}
\lim_{t \to \infty} \mathbb{E}_{x,y}[X_t^{n}Y_t^{m}] = \frac{(\bc K+\bz)x+(c+z)y}{c+\bc K+z+\bz}.
\end{equation}
\end{prop}
\begin{proof}
Let $(N_t,M_t)_{t \ge 0}$ be as in Definition \ref{Block}, started in $(n, m)\in \N_0\times\N_0\setminus{(0,0)}$. 
Let $T$ be the first time at which there is only one particle left in the system $(N_t, M_t)_{t \ge 0}$, that is,
$$
T := \inf\big\{t >0: N_t+M_t=1\big\}.
$$
Note that for any finite initial configuration $(n,m)$, the stopping time $T$ has finite expectation. Now, by Theorem \ref{thm:dual},
\begin{align*}
\lim_{t\rightarrow \infty}\E_{x,y}\big[X_t^{n}Y_t^{m}\big]	& = \lim_{t\rightarrow \infty} \E^{n,m}\left[x^{N_t}y^{M_t}\right] \\
			& = \lim_{t\rightarrow \infty} \Big(x \P^{n,m}\big(N_t=1\big)+ y 				\P^{n,m}\big(M_t=1\big)\Big) \\
					& = \frac{x(\bc K+\bz)}{c+\bc K+z+\bz} + \frac{y(c+z)}{c+\alpha K+z+\bz},
\end{align*}
where the last equality holds by convergence to the invariant distribution of a single particle, jumping between the two states `active' and `dormant' at rate $c+z$ resp.\ $\bc K+\bz$, which is given by $(\frac{\bc K+\bz}{c+\bc K+z+\bz)},\frac{c+z}{c+\bc K+z+\bz)}),$ and independent of the choice of $n,m$.
\end{proof}

\begin{corollary}[Fixation in law]\label{cor:fix_law}
In model A, given $c, K$, $(X_t, Y_t)_{t\geq 0}$ converges in distribution as $t\to\infty$ to a two-dimensional random variable $(X_\infty, Y_\infty),$ whose distribution is given by
\begin{equation}
\label{eq:momentconvergence}
 \mathcal{L}_{(x,y)}\big( X_\infty, Y_\infty \big) =  \frac{(\bc K+\bz)x+(c+z)y}{c+\bc K+z+\bz} \delta_{(1,1)} + 
\big(1-\frac{(\bc K+\bz)x+(c+z)y}{c+\bc K+z+\bz}\big) \delta_{(0,0)}.
\end{equation}
\end{corollary}

\begin{proof}
It is easy to see that the only two-dimensional distribution on $[0,1]^2$, for which all moments are constant equal to $\frac{(\bc K+\bz)x+(c+z)y}{c+\bc K+z+\bz}$, is given by
$$
\frac{(\bc K+\bz)x+(c+z)y}{c+\bc K+z+\bz} \delta_{(1,1)} + 
\Big(1-\frac{(\bc K+\bz)x+(c+z)y}{c+\bc K+z+\bz}\Big)\delta_{(0,0)}.
$$
Indeed, uniqueness follows from the moment problem, which is uniquely solvable on $[0,1]^2.$ Convergence in law follows from convergence of all moments due to Theorem 3.3.1 in \cite{EK86} and the Stone-Weierstra\ss\ Theorem.
\end{proof}

\section{The seed bank coalescent with simultaneous switching}
\label{sn:seedbank_coalescent}
We now analyse the backward in time process in more detail. First, we give a formal construction of the the seed bank coalescent with simultaneous switching in terms of marked partitions. For $k \ge 1$, let $\mathcal{P}_k$ be the set of partitions of $\{1,..,k\}$. For $\pi \in \mathcal{P}_k$ let $|\pi|$ be the number of blocks of the partition $\pi.$ We define the space of \emph{marked} partitions to be 
\begin{align}\label{eq:markedpartitions}
\mathcal{P}^{\{a,d\}}_k=\Big\{ (\zeta, \vec{u}) \mid \zeta\in \mathcal{P}_k, \vec{u} \in \{a,d\}^{|\zeta|}\Big\}.
\end{align}
This enables us to attach to each partition block a flag which can be either `active' or `dormant' ($a$ or $d$), so that we can trace whether an ancestral line is currently in the active or dormant part of the population.

Consider two marked partitions $\pi,  \pi^\prime \in \mathcal{P}_k^{\{a,p\}}$, we write $\pi\succ \pi^\prime$ if $\pi^\prime$ can be constructed by merging exactly 2 blocks of $\pi$ carrying the $a$-flag, and the resulting block in ${\bf \pi}^\prime$ obtained from the merging both again carries an $a$-flag.

We use the notation ${\bf \pi}\Join_k {\bf \pi}^\prime$ if ${\bf \pi}^\prime$ can be constructed by changing the flag of precisely $k$ blocks of $\pi$ from $a$ to $d,$ and ${\bf \pi}\Join^l {\bf \pi}^\prime$ if ${\bf \pi}^\prime$ can be constructed by changing the flag of precisely $l$ blocks of $\pi$ from $d$ to $a$.

\begin{defn}[The seed bank $k$-coalescent with simultaneous switching]
\label{defn:k_seed bank_coalescent}
Fix $c,\bc, K,\in (0,\infty)$ and finite measures $\Lambda,\overline{\Lambda}$ on $[0,1]$ such that $\Lambda(\{0\})=\overline{\Lambda}(\{0\})=0.$ For $k\geq 1$ we define the \emph{seed bank $k$-coalescent with simultaneous switching} $(\Pi^{(k)}_t)_{t \ge 0}$ to be the continuous time Markov chain with values in $\mathcal{P}_k^{\{a,d\}}$, characterised by the following transitions:  
\begin{align}
\label{eq:coalescent_transitions}
&{\bf \pi} \mapsto {\bf \pi}^\prime \text{ at rate } \begin{cases}
                                         1  & \text{ if  } \pi \succ \pi'
                                         \\
                                         c+\int_{[0,1]}z(1-z)^{|\pi|}\frac{\Lambda(\dd z) }{z} & \text{ if } \pi\Join_1 \pi' \\
                                         \bc K+\int_{[0,1]}z(1-z)^{|\pi|}\frac{\bLambda(\dd z)}{z} & \text{ if } \pi\Join^1 \pi'\\
                                         \int_{[0,1]}z^k(1-z)^{|\pi|-k} \frac{\Lambda(\dd z)}{z} & \text{ if } \pi\Join_k \pi', 2\leq k\leq |\pi|, \\
                                         \int_{[0,1]}z^l(1-z)^{|\pi|-l}\frac{\bLambda(\dd z)}{z}  & \text{ if } \pi\Join^l \pi', 2\leq l\leq |\pi|.\\
                                        \end{cases}
\end{align}
\end{defn}

\begin{rem}\label{rem:block-coal} Clearly, the block counting process of this coalescent is the same as in Definition \ref{Block}, cf.\ Figure \ref{fig:coal}. Observe the relations $\mu=z^{-1}\Lambda$ and $\bmu=z^{-1}\bLambda$, see also remark \ref{rem:Lambdas}. In this section we will work with $\Lambda$ instead of $\mu,$ which is more convenient in the proofs of the main results.
\end{rem}

\begin{figure}
\begin{center}
\includegraphics[scale=0.32
]{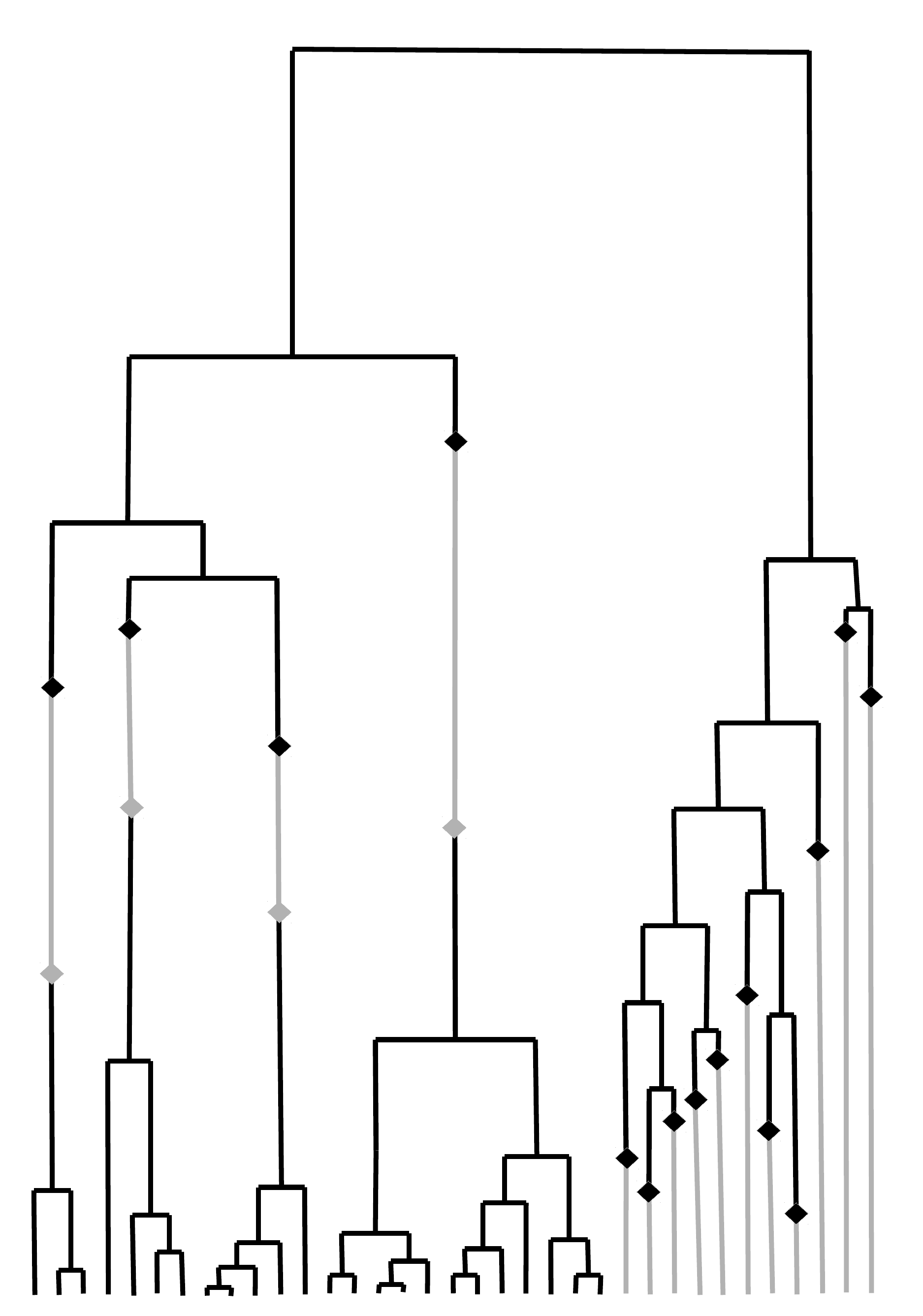}
\hspace{0.5cm}
\includegraphics[scale=0.32
]{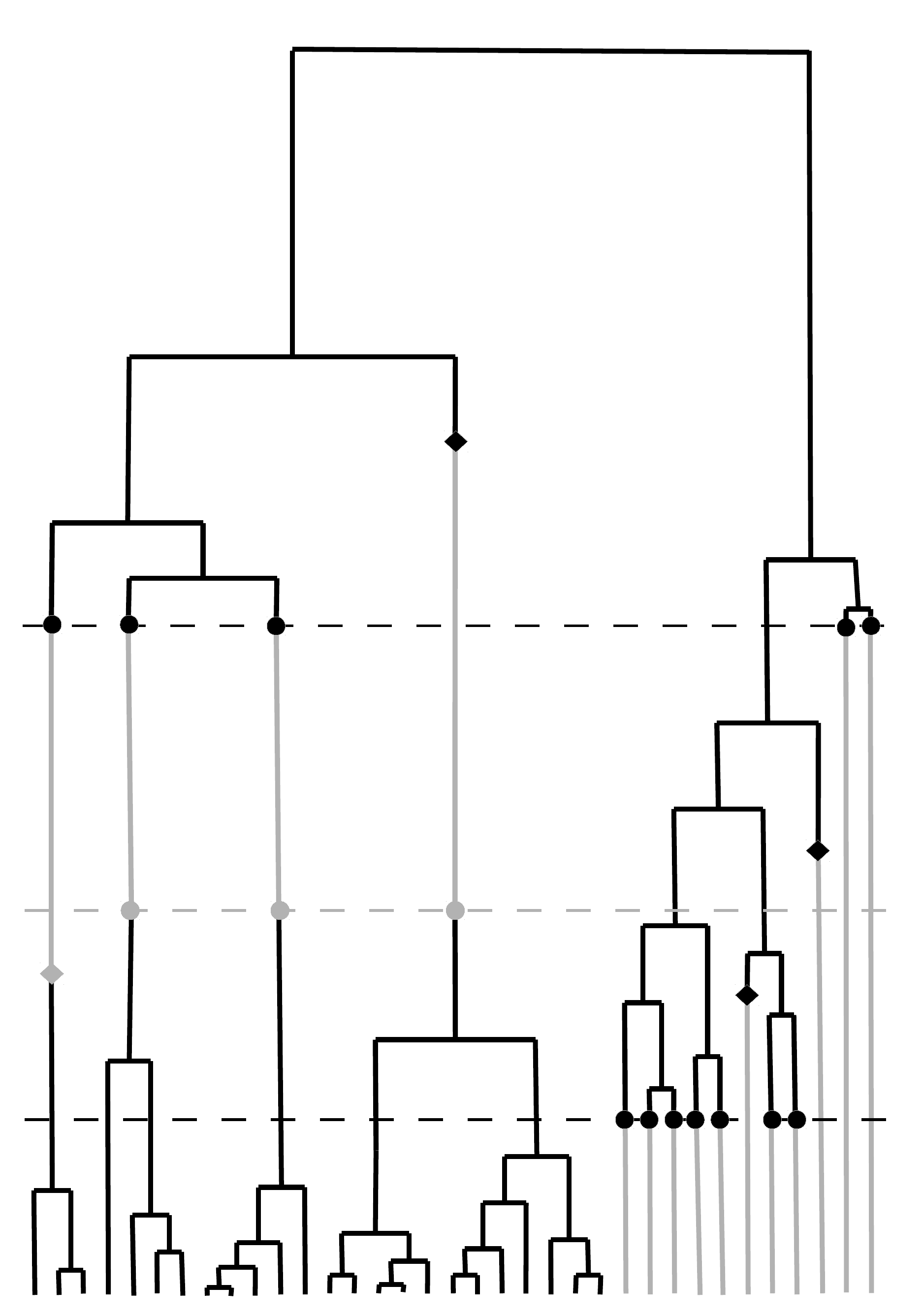}
\caption{Seedbank coalescent with spontaneous switching (left) and with both spontaneous and simultaneous switching (right). Black lines are active, grey lines dormant, diamons indindicate small migration events, dashed lines large migration events.}
\label{fig:coal}
\end{center}
\end{figure}

\begin{defn}[The seed bank coalescent with simultaneous switching]
\label{defn:projective_limit}
We define the \emph{seed bank coalescent with large migration events}, $(\Pi_t)_{t \ge 0}=(\Pi^{(\infty)}_t)_{t \ge 0}$ with intensities $c, \bc$, relative seed bank size $1/K$ and migration measures $\Lambda, \bLambda$ as the unique Markov process obtained as the projective limit as $k$ goes to infinity of the laws of the seed bank $k$-coalescents with simultaneous switching.
\end{defn}

Proving the existence of $(\Pi^{(\infty)}_t)_{t \ge 0}$ via projective limits is standard (the only slightly tedious piece of work is to show that the Markov property is retained under taking projections), which we therefore omit. Note that we are thus allowed to start the block counting process $(N_t,M_t)_{t\geq 0}$ in any state $(n,m),n,m\in\N_0\cup\{\infty\}.$

By entirely similar arguments as the ones presented in Section 3.1 of \cite{BGCKW16} one sees easily that the seed bank coalescent with simultaneous switching is indeed the ancestral process of the seed bank model with simultaneous switching.

For convenience, we give a different construction of the seed bank coalescent with simultaneous switching which will facilitate rigorous proofs in the following section as it allows e.g.\ for simple but precise constructions of couplings. For this construction, we introduce a family (or families) of Poisson point processes (PPP)
\begin{defn}
\label{def:PPPcoalescent}
Let for $c,\bc\in (0,\infty),$ and $\Lambda, \bLambda$ finite measures on $[0,1]$ with $\Lambda(\{0\})=\bLambda(\{0\})=0$
\begin{itemize}
 \item $R_{i_1, i_2}$, $i_1, i_2 \in \N_0$ be a family of PPP on $\R$ with intensity $1\lambda(\dd t)$,
 \item $R_i$, $i \in \N_0$ be a family of PPP on $\R$ with intensity $c\lambda(\dd t)$,
 \item $\bar R_i$, $i \in \N_0$ be a family of PPP on $\R$ with intensity $\bar c\lambda(\dd t)$,
 \item $R_{\Lambda}$ be a PPP on $\R\times[0,1]$ with intensity $\lambda(\dd t)\otimes\frac{\Lambda(\dd z)}{z}$ and
 \item $R_{\bar\Lambda}$ be a PPP on $\R\times[0,1]$ with intensity $\lambda(\dd t)\otimes \frac{\bLambda(\dd z)}{z}$.
\end{itemize}
Here, $\lambda$ denotes the Lebesgue-measure on $\R$.
\end{defn}

\begin{rem}\label{rem:singularity}
Note that we require that the order of the singularity at zero of the intensity measure $\frac{\Lambda(\dd z)}{z}$ is at most of order $z$. This may at first glance look surprising, since this is different from the condition on the singularity for the intensity measure driving a classical $\Lambda$-coalescent, which is of order $z^2$ \cite{P99}. However, a similar condition has been identified in the context of spatial $\Lambda$-Fleming-Viot processes (see \cite{BEV10}), where the authors use a Poisson process with an intensity with a singularity of order $z$ at zero to model large-scale extinction and recolonisation events. Note that one can interpret these events as simultaneous migration of ancestral lines. The point is that we observe singularities of order $z$ every time that we model ``actions'' of a single ancestral line (such as migration, mutation, selective events/branching), and of order $z^2$ in the case of ``actions'' that require more than one ancestral line (such as coalescence). This is a straightforward consequence of interpreting the restriction on the order of the singularity as a minimal condition for the total rate of the dual process to be finite. 
\end{rem}

From these objects, the following characterisation of the seed bank coalescent with simultaneous switching is evident:

\begin{prop}\label{prop:poisson}[Poisson-Point-representation of the coalescent]
Let $\mathtt{p} \in \mathcal{P}^{\{a,d\}}$ the space of marked partitions defined in \eqref{eq:markedpartitions}. The seed bank coalescent with simultaneous switching $\Pi$ is a function of the PPPs given above in the following way: Set $\Pi(0):=\mathtt{p}$. If $t \in \R$ is a (random) time-point in
\begin{itemize}
 \item $R_{i_1, i_2}$: If $i_1$ and $i_2$ are the smallest integers in their respective blocks and both blocks have an $a$-flag in $\Pi(t-)$, then $\Pi(t)$ is the partition where these two blocks are merged and all other blocks remain the same. Otherwise $\Pi(t)=\Pi(t-)$. 
 \item $R_i$: If $i$ is the smallest integer in its block and this has an $a$-flag in $\Pi(t-)$, then $\Pi(t)$ is the partition where this block has a $d$-flag and all other blocks remain the same. 
 \item $\bar R_i$: If $i$ is  the smallest integer in its block and this has a $d$-flag in $\Pi(t-)$, then $\Pi(t)$ is the partition where this block has an $a$-flag and all other blocks remain the same. 
 \end{itemize}
 
 If $(t,z)$ is a point in
 \begin{itemize}
 \item $R_{\Lambda}$, let $\mathbf{u}=({u}_i)_{i\in\N_0}$ be a sequence of independent uniform random variables on $[0,1],$ chosen independently of everything else and independent for each time point. Then $\Pi(t)$ is the partition where all the blocks that had an $a$-flag in $\Pi(t-)$ and whose smallest integer $i^*$ fulfilled ${u}_{i^*} \leq z$ have a $d$-flag while all others remain unchanged. 
 \item $R_{\bar\Lambda}$, let $\mathbf{u}=({u}_i)_{i\in\N_0}$ be a sequence of independent uniform random variables on $[0,1],$ chosen independently of everything else and independent for each time point. Then $\Pi(t)$ is the partition where all the blocks that had a $d$-flag in $\Pi(t-)$ and whose smallest integer $i^*$ fulfilled ${u}_{i^*} \leq z$ have an $a$-flag while all others remain unchanged. 
\end{itemize}
\end{prop}

\subsection{Coming down from infinity}
\label{ssn:coming_down}
The notion of \emph{coming down from infinity} for exchangeable coalescents was introduced by Pitman \cite{P99} and Schweinsberg \cite{S00}. 
They say that the block-counting process $(N_t)_{t \ge 0}$ of a coalescent ``comes down from infinity'', if $N_0=\infty$ $\P$-a.s. and
$$
\P(N_t < \infty)=1 \quad \mbox{ for all } t>0.
$$
They further say that the process ``stays infinite'', if $\P(N_t=\infty)=1$ for all $t\geq 0.$ Note that this leaves intermediate regimes: For example, the ``star-shaped coalescent'' with rates driven by $\delta_1$ has infinitely-many lines until an exp(1)-distributed random time, by which it jumps to a single line only. It thus {\em does} come down from infinity in a certain sense, but only after a strictly positive (random) time. Hence one might want to distinguish between ``coming down from infinity instantaneously'' (Pitman's original definition), ``coming down from infinity after a finite time'', and ``staying infinite''. We mention this because our results regarding the seed bank coalescent with simultaneous switching exhibits all three regimes, as we will see below.

In \cite{BGCKW16} it was proved that the seed bank coalescent does not come down from infinity (neither instantaneously nor after a finite time), due to the fact that even within a very short time, infinitely many lines escape to the seed bank, from where it takes long to come back and be able to coalesce. It turns out that in the case of simultaneous switching, there is a qualitatively different behaviour. 

\begin{thm}[(Not) coming down from infinity]
\label{thm:MRcomingdown}
Assume model B. Let $Y$ be a random variable with distribution $\frac{1}{\Lambda([0,1])}\Lambda$. 
\begin{itemize}
\item[(a)] If $\bLambda(\{1\})=0$, then the process started in $(n,\infty), n\in \N_0\cup\{\infty\}$ will stay infinite.
\item[(b)] If the process is started in $(\infty, m), m\in\N_0,$ then the process comes down from infinity instantaneously if $\E[-\log(Y)] < \infty$ and $c=0$. If $\E[-\log(Y)] =\infty$ or $c>0,$ it stays infinite.
\item[(c)] If $\bLambda(\{1\})>0, c=0$ and $\E[-\log(Y)] < \infty,$ then the process started from $(n,\infty), n\in\N_0\cup \{\infty\}$ comes down from infinity after a finite time, but not instantaneously. 
\end{itemize} 
\end{thm}

\begin{rem}
A finite measure $\Lambda$ on $[0,1]$ with $\E[-\log Y]=\infty$ is for example given by the measure which has density $1_{]0,1/2]} (\log x)^{-2}x^{-1}$ with respect to the Lebesgue measure, which has total mass $1/\log 2.$
\end{rem}

In order to prepare the proof of part (a), we formulate and prove the result for two special cases.

\begin{lemma}\label{lem:cdi}
\begin{itemize}
\item[(i)] If $z^{-1}\bLambda(\dd z)$ is a finite measure on $[0,1]$ and $\bLambda(\{1\})=0$, then the process started in $(n,\infty), n\in \N_0\cup\{\infty\}$ will stay infinite.
\item[(ii)] If there exists $0<\delta<1$ such that $\bLambda([\delta,1])=0$, then the process started in $(n,\infty), n\in \N_0\cup\{\infty\}$ will stay infinite.
\end{itemize}
\end{lemma}

\begin{proof}
Proof of (i). Let $R_{\bLambda}$ be defined as in Definition \ref{def:PPPcoalescent}. In order to turn this into a stochastic process, we consider the first component as time axis and represent the PPP by its (a.s.) finite collection of atoms $\{(t_i, \bar{Z}_{t_i})\}_{i \in \N}$ ordered in time-increasing fashion, which is possible due to the finiteness assumption on the measure $z^{-1}\bLambda(\dd z).$ We introduce its canonical filtration by letting
\be
\mathcal{F}_t:= \sigma \big\{ (t_i, \bar{Z}_{t_i}), 0 \le t_i \le t\big\}, \quad t \ge 0. 
\ee
Fix $t\geq 0.$ Denote by 
$$
i^*:= \max\{i \, : \, t_i \le t\}
$$
the index of the last atom of the PPP before time $t$ (which is of course random). Further, assume that the probability space $(\Omega, \mathcal{F},\P)$ on which the PPPs are defined is large enough to accommodate a doubly infinite sequence of independent uniform random variables $\mathbf{u}=(u_{i,j})_{i,j \in \N}$ independent of everything else. They are used to determine whether or not a block $j$ participates in the large migration event at time $t_i$ whose size is determined by $Z_{t_i},$ cf. Proposition \ref{prop:poisson}. 

We now assume $\bLambda(\{1\})=0$. Assign labels $j\in \N$ arbitrarily to the infinitely many dormant individuals. Let $B_j =1$ if the $j$th individual never left the seed bank until time $t$, and $B_j=0$ otherwise. Then,
\begin{align*}
\P^{n,\infty}(M_t=\infty)\geq \P^{n,\infty}\Big(\sum_{j=1}^\infty B_j=\infty\Big)
&=\P^{n,\infty}\Big(\limsup_{j\to\infty} B_j=1\Big)\\
&=\E\Big(\P^{n,\infty}\big(\limsup_{j\to\infty} B_j=1\big|\mathcal{F}_t\big)\Big).
\end{align*}
Note that conditionally on $\mathcal{F}_t$, the $\{B_j, j \in\N\}$ are independent Bernoulli random variables. By Borel Cantelli, we are done once we can show that
$$
\P_{n,\infty}(B_j=1|\mathcal{F}_t) >0, \quad \P\mbox{-a.s.}, 
$$
where the probability is random, but independent of the index $j$, hence there is a uniform (in $j$) random lower bound away from 0. 
We have, using measurability of the $\bar{Z}_{t_{i}, i \le i^*}$ wrt.\ $\mathcal{F}_t$ and independence of the $u_{i,j}$ from $\mathcal{F}_t$,
\begin{align*}
\P^{n,\infty}(B_j=1|\mathcal{F}_t) &= e^{-\bc t}\,  \P \Big[ {\bf 1}_{\{u_{1,j}>\bar{Z}_{t_{1}}\}} \cdots {\bf 1}_{\{u_{i^*,j}>\bar{Z}_{t_{i^*}}\}}\Big|\mathcal{F}_t\Big] \\
&= e^{-\bc t} \, \E\Big[ {\bf 1}_{[Z_{t_1},1]}(u_{1,j})\cdots {\bf 1}_{[Z_{t_{i^*}},1]}(u_{i^*,j})\Big] \quad \P\mbox{-a.s.,}
\end{align*}
where the expectation in the second line only acts on the $u_{i,j}, i,j\in \N$. We have $\P(u_{1,j}>{Z}_{t_{1}})>0$ since $\bLambda(\{0\})=0,$ and $\P(i^*< \infty)=1$ due to the assumption that $\bmu$ is finite. Thus the r.h.s above is strictly positve by independence, and the event
$$
\Big\{ \P^{(n,\infty)}(B_j=1|\mathcal{F}_t) >0\Big\} 
$$
has probability 1. 

\medskip

Proof of (ii). We first assume $\bLambda$ is a discrete measure, and proceed similarly to the proof of part (i). Let $\bLambda$ be of the form $\bLambda(\dd z)=z\sum_{i=1}^\infty a_i\delta_{\bz_i}(\dd z),$ with $1>\bz_1\geq \bz_2\geq...> 0, a_i\geq 0$ such that $\sum_{i=1}^\infty \bz_i a_i<\infty.$ Note that the latter condition ensures that $\bLambda$ is a finite measure, while $\sum_{i=1}^\infty a_i$ may be infinite. Fix $t\geq 0.$ As in the proof of $(i),$ let $B_i =1$ if the $i$th seed never left the seed bank until time $t$, and $B_i=0$ otherwise. Denote by $K_t^i$ the number of points of size $\bz_i$ up to time $t$ in the PPP $R_{\bLambda},$ and let $K_t=(K_t^2,K_t^3,...)$. Then again $(B_j)_{j\in\N}$ is a sequence of identically distributed Bernoulli random variables conditionally independent given $K_t$, and
$$\P^{n,\infty}({M}_t=\infty)\geq \P^{n,\infty}(\sum_{j=1}^\infty B_j=\infty)=\E[\P^{n,\infty}(\sum_{j=1}^\infty B_j=\infty\,|\, K_t)]. $$
Thus, as in part (i), by Borel Cantelli we are done if we prove that $\P(\P(B_1=1\,|\, K_t)>0)>0.$ We have
\begin{equation}\label{eq}
\P^{n,\infty}(B_1=1\,|\, K_t)=e^{-\bc t}\prod_{i=1}^\infty (1-\bz_i)^{K_t^i}=e^{-\bc t}\exp(\sum_{i=1}^\infty K_t^i\log(1-\bz_i))
\end{equation} 
Now observe that $\E[K_t^i]=ta_i$, and hence $
\E[\sum_{i=1}^\infty \bz_i^k K_t^i]=\sum_{i=1}^\infty \bz_i^k a_it<\infty$ for $k\geq 1$, since $\sum_{i=1}^\infty a_u\bz_i<\infty.$
By Taylor expansion, this implies $\E[\sum_{i=1}^\infty K_t^i\log(1-\bz_i)]<\infty$ and in particular $\P(\sum_{i=1}^\infty K_t^i\log(1-\bz_i)<\infty)=1.$ Thus \eqref{eq} implies that $\P(\P_{n,\infty}(B_1=1\,|\, K_t)>0)>0,$ and the process stays infinite.

For more general measure $\bLambda$ with $\bLambda([\delta,1])=0$ we use an easy coupling with a discretised measure. Set $\bz_1=\delta,$ and choose arbitrary points $\bz_i, i\geq 2$ such that $\bz_1\geq \bz_2\geq ... > 0$ and $\bz_i/\bz_{i+1}\leq 2$ for all $i\in\N$ (for example, one may choose $\bz_i=1/i, i\geq 2$). 

Let $\hat \Pi$ denote the process constructed using the same point processes from Definition \ref{def:PPPcoalescent} as $\Pi,$ but ignoring all events in $R_{i_1,i_2}, R_i$ and $R_\Lambda,$ and where $\bLambda$ is replaced by the measure $\hat{\Lambda}$ on $\{1/i: i\geq 2\}$ defined via $\hat{\Lambda}(\dd z)=z \sum_{i=1}^\infty a_i \delta_{\bz_i}(\dd z),$ where
 \[a_i:=\int_{\bz_{i+1}}^{\bz_i}\bmu(\dd z).\]
This measure can be interpreted as follows: Whenever a value drawn according to the measure $\bmu(\dd z)=z^{-1}\bLambda(\dd z)$ falls into the interval $[\bz_{i+1},\bz_i),$ then the measure $z^{-1}\hat{\Lambda}$ yields value $\bz_i.$ Thus a seed bank process with measure $\hat{\Lambda}$ makes jumps of larger size than for $\bLambda$ from the seed bank to the plant part. Let $(\hat{N}_t, \hat{M}_t)_{t\geq 0}$ denote the corresponding block counting process. By construction, the processes can be coupled such that $\hat{M}_t\leq M_t$ for all $t\geq 0.$ This implies that if $\hat{M}_t=\infty,$ also $M_t=\infty.$ We have
\[\sum_{i=1}^\infty a_i\bz_i=\sum_{i=1}^\infty \bz_i\int_{\bz_{i+1}}^{\bz_i}\frac{\bLambda(\dd z)}{z}\leq 2\int_0^{\bz_1}\bLambda(\dd z)<\infty.\]
Thus $\hat{M}$ stays infinite, as we have seen in the first case, and we are done.
\end{proof}

\begin{proof}[Proof of Theorem \ref{thm:MRcomingdown}]
 \textbf{Part $(a)$:} Fix $\bLambda$ with $\bLambda(\{1\})=0.$ Let $\bLambda_1:=\bLambda(\cdot \cap [0,1/2])$ and $\bLambda_1:=\bLambda(\cdot \cap (1/2,1))$. Then $z^{-1}\bLambda_2(\dd z)$ is a finite measure by our assumptions, and according to Lemma \ref{lem:cdi} (i) the process with jump measure $\bLambda_2$ stays infinite. Due to part (ii) of the Lemma, also the process with $\bLambda_1$ stays infinite. We have $\bLambda=\bLambda_1+\bLambda_2,$ and the support of the two measures don't intersect. By construction, the Poisson point process $R_{\bLambda}$ is the superposition of two PPPs $R_{\bLambda_1}$ and $R_{\bLambda_2}$ with intensity measures $\lambda(\dd t)\otimes \frac{\bLambda_1(\dd z)}{z}$ resp. $\lambda(\dd t)\otimes \frac{\bLambda_2(\dd z)}{z}$. Fix $t>0.$ Since $z^{-1}\bLambda_2(\dd z)$ is finite, we can order its time points $t_i\leq t$ increasingly. At $t_1-$ the process $M_t$ is infinite almost surely because only events of $R_{\bLambda_1}$ have happened before, and it stays infinite at $t_1$. Thus by the strong Markov property $M_{t_i}$ is almost surely infinite for every $t_i,$ and hence $\P(M_t=\infty)=1.$

 \textbf{Part $(b)$:} If in (b) there is $c>0,$ then it can be seen by following the proof of Theorem 4.1 in \cite{BGCKW16}, the process does not come down from infinity, since a sufficiently large number of blocks move immediately from active to dormant. Thus we are in the situation of (a), at least if $\bLambda(\{1\})=0.$ The case $\bLambda(\{1\})>0$ will be discussed at the end of this proof. Assume now $c=0$. We will consider auxiliary processes with helpful properties: Let $(\tilde N, \tilde M)$ be  the process with the same mechanism of coalescence and migration from active to dormant, given by $\Lambda$, but without any migration from dormant to active. For a formal definition of the process, use the construction of the seed bank coalescent via Poisson point processes provided in Proposition \ref{prop:poisson}, using the same $R_{i_1, i_2}$, $i_1, i_2 \in \N$ and $R_{\Lambda}$ as for the original process, but ignoring all other events. $(\tilde N, \tilde M)$ has the essential mechanism we want to analyze in this part. But, as we will discuss in remark \ref{rem:canUseLambdaAuxiliary} below, it is not yet the suitable object for calculations. We will instead work with the process $(\tilde N^{(1)}, \tilde M^{(1)}),$ which has the same transitions as $(\tilde N, \tilde M)$, but at any large migration event determined by the points of $R_{\Lambda}$ only one blocks moves to the seed bank. For a formal definition of this process, we use once more the Poisson construction: Let $(\tilde \Pi^{(1)})$ be the process using the exact same PPPs $R_{i_1, i_2}$, $i_1, i_2 \in \N$ and a slightly different mechanism for the events in $R_{\Lambda}$ (whilst still ignoring all other events/PPPs): If $(t,z) \in \R$ is a point in $R_{\Lambda}$ and $({u})$ a sequence of iid uniform random variables on $[0,1]$ as in Proposition \ref{prop:poisson}, then $\tilde\Pi^{(1)}(t)$ is the partition where the block with an $a$-flag containing the smallest integer $i^*$ among all the blocks with an $a$-flag in $\tilde\Pi^{(1)}(t-)$ fullfilling ${u}_{i^*} \leq z$ has a $d$-flag while all others remain unchanged. In other words, every time $R_{\Lambda}$ commands a (possibly large) migration from active to dormant, $(\tilde \Pi^{(1)})$ will only let the line with the smallest integer migrate. Therefore $(\tilde N^{(1)}, \tilde M^{(1)})$ its block counting process will only have jumps of size 1 at a frequency determined by $R_{\Lambda}$.
 
 In order to proceed, define the stopping times
 \begin{align*}
 \tilde\tau_n &:= \inf\{ t > 0 \mid \tilde N^{(1)}(t) \leq n\}, \; \text{ for } n \in \N_0.
\end{align*}
 As we saw above for $( N,  M)$, we can easily couple $(\tilde N^{(1)}, \tilde M^{(1)})$ to the \emph{Kingman-coalescent} $(\Pi^{\mathcal K})$ in a way that, if we denote by $(\mathcal K)$ the blockcounting process of the Kingman-coalescent, we have 
\begin{align*}
 \P^{\infty,m_0}\left(\forall \; t \geq 0: \; \tilde N^{(1)}(t) \leq \mathcal K(t)\right) = 1.
\end{align*}
This immediately implies 
\begin{align*}
 \E^{\infty,m_0}\left[\tilde\tau_n\right] \leq 2
\end{align*}
 for all  $n \in \N_0$, since the right-hand-side is the \emph{time to the most recent ancestor} in the Kingman-coalescent. Let $\tilde{\mathcal{M}}$ and $\tilde{\mathcal{M}}^{(1)}$ denote the \emph{total} number of lines that migrated from active to dormant in $(\tilde\Pi)$, resp. $(\tilde\Pi^{(1)})$, at any point in time and let $\mathtt{mig}(n)$ be the event that there was a migration (not a coalescence) at time $\tilde\tau_n$, $n \in \N_0$. A moment of thought reveals that the coupling between the two processes implies the estimates
 \begin{align}\label{eq:coupling_estimates}
  \sum_{n \in \N_0} \1_{\mathtt{mig}(n)} \quad  = \quad \tilde{\mathcal{M}}^{(1)} \quad \leq \quad \tilde{\mathcal{M}} \quad \leq \quad  \sum_{n \in \N_0} \underbrace{\tilde N^{(1)}(\tau_n-)}_{\;=\; n+1}\1_{\mathtt{mig}(n)}. 
 \end{align}
Hence, the number of lines that found its way into the seed bank depends on how many of the events $\mathtt{mig}(n)$ are realized. Observe that, since a migration and a colaescence result in jumps of the same size, the events $\mathtt{mig}(n)$, $n \in \N_0$ are actually independent and we can calculate their probability. 
Define $\gamma(n)$ as the rate at which \emph{any} migration event occurs given we have $n \in \N$ active (in $(\tilde \Pi^{(1)})$). Then
\begin{align}\label{eq:gammas}
\gamma(n)	& = \sum_{j=1}^{n}\int_{[0,1]}\binom{n}{j}z^j(1-z)^{n-j}\frac{\Lambda(\dd z)}{z}\notag\\
		& = \int_{[0,1]} 1-(1-z)^n\frac{\Lambda(\dd z)}{z\notag}\\
		& = \int_{[0,1]} \int_{[0,1]} - \frac{\dd}{\dd u}(1-uz)^ndu\frac{\Lambda(\dd z)}{z}\notag \\
		& = \int_{[0,1]} \int_{[0,1]} n(1-uz)^{n-1}\dd u\Lambda(\dd z)\notag\\
		& = n \Lambda[0,1] \E [(1-W)^{n-1}] 
\end{align}
where $W:=UY$ for a random variable $U$ uniformly distributed on $[0,1]$ independent of $Y\sim \Lambda([0,1])^{-1}\Lambda$. The last equality is inspired by a represantation in Theorem 2 of \cite{G14}. Therefore
\begin{align*}
\P^{\infty,m_0}(\mathtt{mig}(n-1))	& = \frac{\gamma(n)}{\binom{n}{2}+\gamma(n)}\\
					& =\frac{\Lambda[0,1]\E[(1-W)^{n-1}]}{\frac{n-1}{2}+\E[(1-W)^{n-1}]} = \frac{2\Lambda[0,1]}{n-1}\E[(1-W)^{n-1}] + O(n^{-2})
\end{align*}
for any $n\geq 2,$ which in turn implies
\begin{align*}
\sum_{n=1}^\infty \P^{\infty,m_0}(\mathtt{mig}(n))	& = \; \sum_{n=1}^\infty \frac{2\Lambda[0,1]}{n}\E\left[(1-W)^{n}\right] + \text{const}\\
					& = \; 2\Lambda[0,1] \E[-\log(W)] + \text{const}. 
\end{align*}

Borel-Cantelli gives that almost surely only finitely many of the events $\mathtt{mig}(n),$ for $n\in\N$ happen, and thus both sums in \eqref{eq:coupling_estimates} are finite, if and only if $\E[-\log(W)] < \infty$. Observe that $\E[-\log(W)] = \E[-\log(U)]+\E[-\log(Y)]$ is finite if and only if $\E[-\log(Y)]<\infty.$ Since Kingman's coalescent comes down from infinity instantaneously, we see that $(N_t, M_t)$ comes down from infinity immediately if $\E[-\log Y]<\infty.$ Otherwise, the process stays infinite, at least provided $\bLambda(\{1\})=0,$ since in that case by (a) infinitely many blocks migrating to the seed bank in an arbitrarily short time implies that the process stays infinite. If $\bLambda(\{1\})>0,$ there is a positive probability for all dormant blocks to become active at the same time. In that case, the process starts afresh from $(\infty, 0).$ By the strong Markov property, and the above proof, we will again have infinitely many blocks moving to the seed bank, and thus the process will stay infinite also in this case.

 \textbf{Part $(c)$:} As we just argued in the last lines of the proof of part (b), if $\bLambda(\{1\})>0,$ there is a positive probability for all dormant blocks moving to the active part at once. By Borel Cantelli, this event eventually happens with probability one, and thus by (b) the process comes down from infinity. However, the coming down only happens after the seed bank has been emptied, and not instantaneously.
  \end{proof}

 \begin{rem}\label{rem:canUseLambdaAuxiliary}
  One might think it easier (or more precise) to estimate the number $\tilde{\mathcal{M}}$ of lines that migrated from active to dormant in $(\tilde\Pi)$ with the help of $(\tilde N, \tilde M)$ directly. With the same idea we can define stopping times $\tau_n:=\inf\{t>0 \mid \tilde N(t) \leq n\}$. Since we have large jumps, these may coincide for several $n \in \N$ so one is tempted to define the random times $\tau^*_k:= \sup\{t < \tau_k \mid \tilde N(t-) \neq \tilde N(t)\} = \sup\{ t < \tau_k \mid \exists \; n \in \N_0: \; t = \tau_n\}$ as the actual jump times. There is a small difficulty in that we cannot bound the value of $\tilde N$ at any such time, but much worse is that we actually cannot calculate the probabilities of the events ``there is a migration at $\tau_n$'' or ``there is a migration at $\tau^*_k$'' as before. Indeed, both $\tau_n$ and $\tau^*_k$ contain information about the present and the future of the process and therefore the latter are not even stopping times.
 \end{rem}

\paragraph{Acknowledgements.} The authors acknowledge support by the DFG Priority Programme SPP 1590 ``Probabilistic Structures in Evolution'', grants no. BL 1105/5-1 and KU 2886/1-1. ACG was supported by grant no. UNAM PAPIIT IA100419. Part of this work was completed while AGC was a BMS Substitute Professor at TU Berlin supported by the Berlin Mathematical School.

\bibliography{Bib_WFdiffusion.bib}

\newcommand{\etalchar}[1]{$^{#1}$}
\begin{thebibliography}{BGCKWB16}

\bibitem[BBGW17]{BBGCWB18+}
J.~{Blath}, E.~{Buzzoni}, A.~{Gonz{\'a}lez Casanova}, and M.~{Wilke-Berenguer}.
\newblock {Structural properties of the seed bank and the two-island
  diffusion}.
\newblock {\em ArXiv e-prints}, October 2017.

\bibitem[BBKW18]{BBKWB18+}
J.~{Blath}, E.~{Buzzoni}, J.~{Koskela}, and M.~{Wilke-Berenguer}.
\newblock {Inference for seed bank coalescents}.
\newblock {\em Preprint}, 2018.

\bibitem[BEV10]{BEV10}
N.~H. Barton, A.~M. Etheridge, and A.~V\'{e}ber.
\newblock A new model for evolution in a spatial continuum.
\newblock {\em Electron. J. Probab.}, 15:no. 7, 162--216, 2010.

\bibitem[BGCE{\etalchar{+}}15]{BEGCKW15}
J.~Blath, A.~Gonz\'alez~Casanova, B.~Eldon, N.~Kurt, and M.~Wilke-Berenguer.
\newblock Genetic {V}ariability under the {S}eedbank {C}oalescent.
\newblock {\em Genetics}, 200(3):921--934, 2015.

\bibitem[BGCKWB16]{BGCKW16}
J.~Blath, A.~Gonz\'alez~Casanova, N.~Kurt, and M.~Wilke-Berenguer.
\newblock A new coalescent for seed-bank models.
\newblock {\em Ann. Appl. Probab.}, 26(2):857--891, 2016.

\bibitem[BGKS13]{BGCKS12}
J.~Blath, A.~{Gonz{\'a}lez Casanova}, N.~Kurt, and D.~Span{\`o}.
\newblock The ancestral process of long-range seed bank models.
\newblock {\em J. Appl. Probab.}, 50(3):741--759, 2013.

\bibitem[BLP15]{BLP15}
M\'{a}ty\'{a}s Barczy, Zenghu Li, and Gyula Pap.
\newblock Yamada-{W}atanabe results for stochastic differential equations with
  jumps.
\newblock {\em Int. J. Stoch. Anal.}, pages Art. ID 460472, 23, 2015.

\bibitem[dHP17]{DHP16}
F.~den Hollander and G.~Pederzani.
\newblock Multi-colony {W}right-{F}isher with seed-bank.
\newblock {\em Indagationes Mathematicae}, 28(3):637 -- 669, 2017.

\bibitem[DK99]{DK99}
Peter Donnelly and Thomas~G. Kurtz.
\newblock Particle representations for measure-valued population models.
\newblock {\em Ann. Probab.}, 27(1):166--205, 1999.

\bibitem[EK86]{EK86}
S.N. Ethier and T.G. Kurtz.
\newblock {\em Markov processes}.
\newblock Wiley Series in Probability and Mathematical Statistics: Probability
  and Mathematical Statistics. John Wiley \& Sons, Inc., New York, 1986.
\newblock Characterization and convergence.

\bibitem[Gri14]{G14}
R.~Griffiths.
\newblock The lambda-fleming viot process and a connection with wright-fisher
  diffusion.
\newblock {\em Adv. Appl. Prob.}, (46):1009 -- 1035, 2014.

\bibitem[Her94]{H94}
H.M. Herbots.
\newblock {\em {Stochastic models in population genetics: genealogical and
  genetic differentiation in structured populations}}.
\newblock PhD thesis, University of London, 1994.

\bibitem[JK14]{JK14}
S.~Jansen and N.~Kurt.
\newblock On the notion(s) of duality for {M}arkov processes.
\newblock {\em Probab. Surv.}, 11:59--120, 2014.

\bibitem[Kal02]{K02}
Olav Kallenberg.
\newblock {\em Foundations of modern probability}.
\newblock Probability and its Applications (New York). Springer-Verlag, New
  York, second edition, 2002.

\bibitem[KKL01]{KKL01}
Ingemar Kaj, Stephen~M. Krone, and Martin Lascoux.
\newblock Coalescent theory for seed bank models.
\newblock {\em J. Appl. Probab.}, 38(2):285--300, 2001.

\bibitem[Kur07]{K07}
Thomas~G. Kurtz.
\newblock The {Y}amada-{W}atanabe-{E}ngelbert theorem for general stochastic
  equations and inequalities.
\newblock {\em Electron. J. Probab.}, 12:951--965, 2007.

\bibitem[Kur14]{K14}
Thomas~G. Kurtz.
\newblock Weak and strong solutions of general stochastic models.
\newblock {\em Electron. Commun. Probab.}, 19:no. 58, 16, 2014.

\bibitem[KZH08]{KZH08}
A.R.R. Kermany, X.~Zhou, and D.A. Hickey.
\newblock Joint stationary moments of a two-island diffusion model of
  population subdivision.
\newblock {\em Theoretical Population Biology}, 74(3):226--232, 2008.

\bibitem[LJ11]{LJ11}
J.~T. Lennon and S.~E. Jones.
\newblock Microbial seed banks: the ecological and evolutionary implications of
  dormancy.
\newblock {\em Nat. Rev. Microbiol.}, 9(2):119--130, 2011.

\bibitem[LM15]{LM15}
A.~Lambert and C.~Ma.
\newblock The coalescent in peripatric metapopulations.
\newblock {\em J. Appl. Probab.}, 52(2):538--557, 2015.

\bibitem[MKAv17]{MKTZ17}
J.~M\"uller, B.~Koopmann, A.~Tellier A., and D.~\v{Z}ivkovi\v{c}.
\newblock Fisher-wright model with deterministic seed bank and selection.
\newblock {\em Theor. Pop. Biol.}, pages 29–--39, 2017.

\bibitem[Not90]{N90}
M.~Notohara.
\newblock The coalescent and the genealogical process in geographically
  structured population.
\newblock {\em Journal of Mathematical Biology}, 29(1):59--75, 1990.

\bibitem[Pit99]{P99}
J.~Pitman.
\newblock Coalescents with multiple collisions.
\newblock {\em Ann. Probab.}, 27(4):1870--1902, 1999.

\bibitem[Sag99]{S99}
S.~Sagitov.
\newblock The general coalescent with asynchronous mergers of ancestral lines.
\newblock {\em J. Appl. Probab.}, 36(4):1116--1125, 1999.

\bibitem[Sch00]{S00}
Jason Schweinsberg.
\newblock A necessary and sufficient condition for the {$\Lambda$}-coalescent
  to come down from infinity.
\newblock {\em Electron. Comm. Probab.}, 5:1--11, 2000.

\bibitem[SL17]{SL18}
William~R. Shoemaker and Jay~T. Lennon.
\newblock Evolution with a seed bank: The population genetic consequences of
  microbial dormancy.
\newblock {\em Evolutionary Applications}, 11(1):60--75, 2017.

\bibitem[Tak88]{T88}
N.~Takahata.
\newblock The coalescent in two partially isolated diffusion populations.
\newblock {\em Genet. Res.}, 53:213--222, 1988.

\bibitem[TLL{\etalchar{+}}11]{T11}
A.~Tellier, S.J.Y. Laurent, H.~Lainer, P.~Pavlidis, , and W.~Stephan.
\newblock Inference of seed bank parameters in two wild tomato species using
  ecological and genetic data.
\newblock {\em PNAS}, 108(41):17052--17057, 2011.

\bibitem[VGO04]{V04}
R.~Vitalis, S.~Gl\'emin, and I.~Olivieri.
\newblock When genes go to sleep: The population genetic consequences of seed
  dormancy and monocarpic perenniality.
\newblock {\em The American Naturalist}, 163(2):259--311, 2004.

\bibitem[Wri51]{W51}
S.~Wright.
\newblock The genetical structure of populations.
\newblock {\em Annals of Eugenics}, 15(1):323--354, 1951.

\bibitem[ZT12]{ZT12}
D.~Zivkovi\v{c} and A.~Tellier.
\newblock Germ banks affect the inference of past demographic events.
\newblock {\em Mol. Ecol.}, 21:5434--5446, 2012.

\end{thebibliography}
\bibliographystyle{alpha}

\end{document}